\documentclass{article}

\usepackage[a4paper]{geometry}
\usepackage{amssymb}
\usepackage{amsmath}
\usepackage{amsfonts}
\usepackage{amsthm}
\usepackage{url}
\usepackage{xspace}
\usepackage[T1]{fontenc}
\usepackage{verbatim}

\newtheorem{definition}{Definition}
\newtheorem{theorem}{Theorem}
\newtheorem{corollary}{Corollary}
\newtheorem{proposition}{Proposition}

\newcommand{\SC}{\textsc{min set cover}\xspace}
\newcommand{\DS}{\textsc{min dominating set}\xspace}
\newcommand{\VC}{\textsc{min vertex cover}\xspace}
\newcommand{\is}{\textsc{max independent set}\xspace}

\newcommand{\IDS}{\textsc{min independent dominating set}\xspace}
\newcommand{\col}{\textsc{min vertex coloring}\xspace}
\newcommand{\fvs}{\textsc{min feedback vertex set}\xspace}

\title{\textbf{Parameterized (in)approximability of subset problems}}
\author{Edouard Bonnet \hspace*{1cm} Vangelis~Th.~Paschos\footnote{Also, Institut Universitaire de France} \\
PSL Research University, Universit\'e Paris-Dauphine \\
LAMSADE, CNRS UMR 7243, France \\
\texttt{edouard.bonnet@dauphine.fr,paschos@lamsade.dauphine.fr}}

\begin{document}

\maketitle

\begin{abstract}

We discuss approximability and inapproximability in FPT-time for a large class of \emph{subset problems} where a feasible solution~$S$ is a subset of the input data and the value of~$S$ is~$|S|$. The class handled encompasses many well-known graph, set, or satisfiability problems such as \DS, \VC, \SC, \is, \fvs, etc. In a first time, we introduce the notion of \emph{intersective} approximability that generalizes the one of \emph{safe} approximability introduced in \cite{safeapprox} 
and show strong parameterized inapproximability results for many of the subset problems handled. Then, we study approximability of these problems with respect to the dual parameter $n-k$ where~$n$ is the size of the instance and~$k$ the standard parameter. More precisely, we show that under such a parameterization, many of these problems, while \textbf{W[$\mathbf{\cdot}$]}-hard, admit parameterized approximation schemata.

\end{abstract}

\section{Introduction}\label{intro}

Parameterized approximation aims at bringing together two very active fields of the theoretical computer science, polynomial approximation and parameterized computation. We say that a minimization (maximization, respectively) problem~$\Pi$,  together with a parameter~$k$, is {\em parameterized $r$-approximable}, if there exists an FPT-time algorithm which computes a solution of size at most (at least, respectively)~$rk$ whenever the input instance has a solution of size at most (at least, respectively) $k$, otherwise, it outputs an arbitrary solution. This line of research was initiated by three independent works~\cite{dofemcciwpec,caihuiwpec,chgrogruiwpec}. For a very interesting overview, see~\cite{marx-approx}. 

Here, we first handle approximability and inapproximability in FPT-time of \emph{subset problems} where a feasible solution~$S$ is a subset of elements encoding the input that verifies some specific property, and the value of~$S$ is~$|S|$. 
This class includes problems as \DS, \VC, \SC, \is and numerous other well-known problems.

In Section~\ref{intersective}, we introduce the notion of \emph{intersective approximation} that generalizes the notion of \emph{safe approximation} defined  in~\cite{safeapprox}. An approximation is said to be safe, if it produces solutions containing an optimal solution. As defined in~\cite{safeapprox}, safe approximation only captures minimization problems and can be used in order to get strong inapproximability results. For instance, it is shown in~\cite{safeapprox} that a safe $c\log{n}$-approximation, for any $c  > 0$, for \DS, would be transformed into an exact FPT algorithm for this problem, contradicting $\textbf{FPT} \neq \textbf{W[2]}$. 

Although very interesting, safe approximation can be proved to be true rather rarely, since only \VC and some restrictive versions of a few minimization subset problems admit such approximations. Intersective approximability relaxes the requirement of inclusion of an optimal solution into the output approximation by just asking these two solutions to have a non-empty intersection. This relaxation makes that this approximability may apply also to maximization subset problems that are proved to have the same behaviour as those of minimization wrt to this new notion of approximability. Then, we use safe approximability, in order to establish meta-theorems producing as corollaries strong negative results for subset problems.

In the second part of the paper we handle \emph{dual parameterization} of subset problems. Given a subset problem~$\Pi$ of size~$n$ parameterized by the standard parameter~$k$ (i.e., the cardinality of its optimal solution), we call dual parameter for~$\Pi$ the parameter $n-k$. Parameterization by the dual parameter has been studied for many classical problems (see, for example,~\cite{cai,dowfel}) and for a lot of them it has been proved that although hard with respect to the standard parameter, they become easy when parameterized by the dual parameter. This is the case, for instance of the famous \col problem that, although not in \textbf{XP} when parameterized by the chromatic number~$\chi$ it is in \textbf{FPT} when parameterized by $n - \chi$~\cite{paramcompendium,dowfel}, or even for \is when parameterized by the size of a minimum vertex cover (this is a folklore result). In the opposite, as we will see, \SC, when parameterized by $m-k$, where~$m$ is the size of the set-system it is \textbf{W[1]}-hard. Apart the fact that dual parameterization is interesting per se, for numerous problems it has also a natural interpretation. For instance, for \col, dual parameter is the number of unused colors (assuming that~$n$ colors feasibly color the vertices of a graph); for \SC this parameter represents the number of unused sets (since~$m$ sets cover the ground set), for \is dual parameter is the size of a minimum vertex cover, etc. In Section~\ref{differential} we establish several interesting approximability results for subset problems parameterized by the dual parameter.

\section{Intersective approximation}\label{intersective}

As we have already mentioned in Section~\ref{intro}, intersective approximability extends the safe approximability of~\cite{safeapprox}, by allowing the approximate solutions computed not to thoroughly contain an optimal solution (for the case of minimization problems) but only to have a non-empty intersection with some optimal solution. 
\begin{definition}\label{intersectiveapprox}
A $\rho$-approximation algorithm~\texttt{A} is said to be \emph{intersective} for a problem~$\Pi$ if, when running on any instance~$I$ of~$\Pi$, it computes a $\rho$-approximate solution~$\mathtt{A}(I)$ and there exists an optimal solution $S_0$ of~$I$ such that $\mathtt{A}(I) \cap S_0 \neq \emptyset$. 
\end{definition}
Note that a safe approximation is a special case of an intersective approximation. Also, from Definition~\ref{intersectiveapprox}, one can see immediately that intersective approximations do not require the problem to minimise its solution. For instance, \textsc{max minimal vertex cover}\footnote{Given a graph~$G$, \textsc{max minimal vertex cover} consists of determining a maximum-cardinality vertex cover that is minimal for inclusion; an optimal solution for this problem is the complement of a minimum independent dominating set.}~\cite{mmvc_waoa} always admits intersective approximations.
\begin{theorem}\label{main}
If a minimization subset problem~$\Pi$ has an FPT intersective $(c \log n)$-approximation for some constant $c > 0$, then~$\Pi$ admits an exact FPT algorithm. The same holds for maximization subset problems without any condition on the value of the intersective approximation ratio.
\end{theorem}
\begin{proof}
Consider first some minimisation problem~$\Pi$, an intersective FPT approximation algorithm~\texttt{A} for~$\Pi$ achieving approximation ratio~$(c \log n)$ and let~$I$ be any instance of~$\Pi$. Compute $S=\mathtt{A}(I,k)$ a safe approximation for~$I$. If $|S| \geqslant ck\log n$ then answer that~$I$ is a NO-instance. Otherwise, denote by~$e_i$, $i = 1, \ldots, |S|$ the elements of~$S$ and define by~$I(e_i)$ the sub-instance of~$I$ where any solution~$S'$ is such that $S' \cup \{e_i\}$ is a solution for~$I$. Branch on $I(e_1),\ldots,I(e_{|S|})$. For all these instances, compute a $c \log n$-approximation and so on. When~$k$ elements have been taken in the solution, stop and return it.

We claim that the best solution found at a leaf of the so constructed branching tree is an optimal solution. Indeed, starting from the root one can, by definition of intersective approximability, move to a child which is conform to an optimal solution.

The branching tree has depth~$k$ since, at each step, one element is added in the solution, and arity bounded by~$ck \log n$. Hence, the number of its nodes is bounded by~$(ck)^k (\log n)^k$. On each node, some FPT computation is done, bounded by say~$f(k)p(n)$. So, the overall complexity is~$O((ck)^k f(k) (\log n)^k p(n))$. Taking into account that~$(\log n)^k$ is~FPT wrt~$k$~\cite{Sloper08}, concludes the first part of the proof dealing with minimization problems.

For maximization problems, just observe that an exhaustive search of all the subsets of an intersectively approximate solution~$S$ takes time $O^*(2^{|S|}) \leqslant 2^k$.~\end{proof}
Based upon Theorem~\ref{main}, the following holds for the intersective approximability of several well-known problems.
\begin{corollary}
Unless $\textbf{FPT}=\textbf{W[2]}$, no FPT intersective $(c \log n)$-approximation exists for either \SC, or \DS, for any constant $c > 0$. Unless $\textbf{FPT}=\textbf{W[1]}$, no FPT intersective approximation exists either for \is, for \textsc{max clique}, or for \fvs{}.
\end{corollary}
Obviously, the same holds for any \textbf{W[t]}-hard subset problem.

The proof of Theorem~\ref{main} gives also some hints for obtaining FPT algorithms for intersectively approximable problems. For instance, for \VC, or for \textsc{max minimal vertex cover} application of the algorithm of Theorem~\ref{main} derives~FPT algorithms with respect to the standard parameter. Of course inclusion of these problems in~\textbf{FPT} is already known, but this could be the case for other problems of still unknown status. 

Intersective approximability can also be extended to several problems that are not subset problems. We just sketch such an extension to coloring problems. A solution for a $k$-coloring can be seen as $k$ sets $S_1,\ldots,S_k$ where~$S_i$ is the set of vertices (or edges) receiving color~$i$. A $\rho$-intersective approximation to a $k$-coloring problem can be defined as an $h$-coloring $S'_1,\ldots,S'_h$ such that there exists an optimal solution $S_1,\ldots,S_k$ with $k \geqslant h/\rho$ and two integers~$i,j$ satisfying $S_i=S'_j$. Under this definition, the following can be proved with a proof similar to that of Theorem~\ref{main}.
\begin{corollary}
If a $k$-coloring problem~$\Pi$ has an FPT intersective $(c \log n)$-approximation for some constant $c > 0$, then~$\Pi$ admits an exact~FPT algorithm. 
\end{corollary} 

\section{Using polynomial approximation to design FPT approximation schemata for W[$\mathbf{\cdot}$]-hard subset problems}\label{differential}

In what follows, given a problem~$\Pi$ with data-size~$n$ and standard parameter~$k$,~D-$\Pi$ denotes~$\Pi$ parameterized by $n-k$. Note that another (sometimes more comprehensive, in particular when dealing with approximation issues) way to see~D-$\Pi$ is to see it as a ``new'' problem having the same data-size and feasibility constraints as~$\Pi$ and whose goal is inverse to that of~$\Pi$, i.e.,~D-$\Pi$'s goal becomes minimization (resp., maximization) when~$\Pi$'s goal is maximization (resp., minimization). Obviously, $\textrm{D-D-}\Pi = \Pi$. Note finally that, as already mentioned in Section~\ref{intro}, for some~$\Pi$'s, the corresponding~D-$\Pi$ versions have natural expressions. For example, for $\Pi = \col$,~D-$\Pi$ becomes the problem of maximizing the number of unused colors; for $\Pi = \SC$,~D-$\Pi$ consists of maximizing the number of unused sets; when $\Pi = \IDS$,~D-$\Pi$ is \textsc{max minimal vertex cover}, etc.

For some hard~$\Pi$'s, the corresponding~D-$\Pi$ versions can be proved to be in~\textbf{FPT}. This is the case, for instance, of \col~\cite{paramcompendium,dowfel} that is not in~\textbf{XP}, or of \IDS. For this latter problem, $n-k$ is the size of a maximum minimal vertex cover which is bigger than a minimum vertex cover which is bigger than the treewidth of the input graph. The fact that \IDS parameterized by treewidth is FPT, concludes inclusion of \IDS parameterized by $n-k$ in~\textbf{FPT}. 

However, for many other problems, this good news is no more true. Consider, for instance, D-\SC where, given a set~$S$ of~$m$ subsets of a universe~$U$ of~$n$ elements, one wishes to find the maximum number of subsets of~$S$ that can be removed in order that the remaining subsets still cover~$U$. Let~$k$ be the standard parameter for D-\SC. 
\begin{proposition}\label{dsc-hardnes}
D-\SC is \textbf{W[1]}-hard when parameterized by~$k$ and \textbf{W[2]}-complete when parameterized by $m-k$.
\end{proposition}
\begin{proof}
The second claim is easy to be proved: $m-k$ for D-\SC is the standard parameter for \SC and given a solution~$S_0$ for the former problem, one can take $S \setminus S_0$ as solution for the latter. 

To prove the \textbf{W[1]}-hardness of D-\SC when parameterized by its standard parameter, just observe that the restriction of this problem to instances where each set represents a vertex and contains all the edges incident to that vertex,
results in a \is problem that is \textbf{W[1]}-complete when parameterized by the standard parameter. 
\end{proof}
Note that similar negative results hold also for:
\begin{itemize}
\item \VC that is \textbf{W[1]}-complete when parameterized by the size of a maximum independent set~\cite{cross-parameterization}
\item \fvs that is also \textbf{W[1]}-complete when parameterized by the dual parameter~\cite{paramcompendium,fernauhdr,Khot:2002:PCF:638032.638041} (the dual parameterization of \fvs is called \textsc{vertex induced forest}).
\end{itemize}
In what follows, we give a sufficient condition under which, given a subset problem~$\Pi$, problem D-$\Pi$ admits an approximation schema parameterized by the parameter ``size of the instance minus standard parameter of  $\Pi$''.
\begin{theorem}\label{diffschema}
Consider a subset problem~$\Pi$ with standard parameter~$k$ and set $k_D = n - k$, where~$n$ is the size of the data describing~$\Pi$. Then:
\begin{enumerate}
\item\label{ena} if~$\Pi$ is a minimization problem and is approximable in polynomial time within ratio at most~$c\log{n}$, for some $c > 0$, problem~D-$\Pi$ parameterized by~$k$ admits a parameterized approximation schema;
\item\label{dyo} if~$\Pi$ is a maximization problem, problem D-$\Pi$ parameterized by~$k$ admits a parameterized approximation schema independently on the polynomial approximation ratio of~$\Pi$.
\end{enumerate}
\end{theorem}
\begin{proof}
In order to prove Item~\ref{ena}, consider a problem~$\Pi$ admitting the conditions of the item, its D-version~D-$\Pi$, a $\rho$-approximation algorithm~\texttt{A} for~$\Pi$ and denote by~$k'$ the cardinality of the solution returned by~\texttt{A}. The complement of this solution is a solution of size $k'_D = n - k'$ for~D-$\Pi$, while the size of the optimum is $k_D = n - k$. So, the approximation ratio guaranteed for~D-$\Pi$ is (recall that~D-$\Pi$ is a maximization problem):
\begin{equation}\label{difratmin}
\frac{k'_D}{k_D} = \frac{n - k'}{n-k} \geqslant \frac{n - \rho k}{n-k}
\end{equation}
Fix some constant $\epsilon > 0$. Then, to make the last fraction in~(\ref{difratmin}) greater than $1 - \epsilon$, it must hold that:
\begin{equation}\label{difratmin2}
n \geqslant \left(\frac{\rho - 1 + \epsilon}{\epsilon}\right)k
\end{equation}
If~$n$ does not satisfies~(\ref{difratmin2}), then the simple~$O^*(2^n)$-time algorithm that builds all the data-subsets and chooses the one that constitutes the best solution, runs in FPT time as far as $\rho \leqslant c\log{n}$ for some constant $c >0$.

Proof of Item~\ref{dyo} is similar. Here~D-$\Pi$ is a minimization problem and we need that:
\begin{equation}\label{difratmax}
\frac{k'_D}{k_D} = \frac{n - k'}{n-k} \leqslant \frac{n - \rho k}{n-k} \leqslant 1 + \epsilon \Longrightarrow
n \geqslant \left(\frac{1 + \rho + \epsilon}{\epsilon}\right)k \Longrightarrow n \geqslant \left(\frac{2}{\epsilon}\right)k
\end{equation}
since $\rho < 1$ and~$\epsilon$ is considered very small. Once again, if~$n$ does not satisfies~(\ref{difratmax}), an exhaustive search requiring~~$O^*(2^n)$-time becomes an FPT-algorithm.~\end{proof}

Observe that \IDS does not meet conditions of Theorem~\ref{diffschema}. Indeed, the best known polynomial time achievable approximation ratio for \IDS is $\Delta+1$, where~$\Delta$ is the maximum degree of the input graph and can be aribatrily larger than~$O(\log{n})$, and it is inapproximable within~$\Delta^{1 - \epsilon}$, for any $\epsilon > 0$ in polynomial time~\cite{halmmis}.

Note finally, that the scope of Theorem~\ref{diffschema} encompasses more problems than subset ones, for instance, problems whose optimal solutions are structures rather than subsets of the input data. Coloring problems are such problems. Of course, the classical \col problem does not meet conditions of Item~\ref{ena}, since it is inapproximable in polynomial time within better than~$n^{1-\epsilon}$, for any $\epsilon > 0$~\cite{zucker}. Moreover, as it is proved in~\cite{dptas}, D-\col is \textbf{APX}-hard. On the other hand, \textsc{min edge coloring} that is polynomially approximable within ratio~4/3~\cite{Vizing64} or even the \textsc{max edge coloring} problem\footnote{Given an edge-weighted graph~$G$, the weight of a color~$M$ (that is a matching of~$G$) is defined as the weight of the ``heaviest'' edge of~$M$ and the objective is to determine a partition of the edges of~$G$ into matchings, minimizing the sum of their weights.} that is approximable within ratio~2 in polynomial time~\cite{Kesselman07}, for both problems considering~$m$ (the size of the edge-set) as size of the input-data.

In any case, Theorem~\ref{diffschema} has a number of interesting corollaries, seen as good news, for several subset problems. The following corollary summarizes some of them.
\begin{corollary}
D-\SC, D-\DS, D-\is, D-\textsc{max clique}, D-\textsc{max set packing}, D-\fvs, D-\textsc{min edge coloring}, D-\textsc{max edge coloring}, parameterized by~$k_D$, admit parameterized approximation schemata. 
\end{corollary}

\section{Discussion}

We tackled parameterized approximability of subset problems, which constitute a very natural and popular class of combinatorial problems. We have first introduced a notion of rather strong approximability, the intersective approximability, and we have proved several strong negative results for the possibility of subset problems to be intersectively approximable. Although this notion importantly relaxes the safe approximability of~\cite{safeapprox}, it still remains quite strong to be able to produce positive approximation results. On the other hand, as safe approximability, and despite its narrowness, its merit is to give new insight in the field of parameterized approximation that is in its beginnings and needs several precision and hypothesis for stabilizing its formal framework. 

Next, we have proposed a systematic approach for approximating subset problems when parameterized by a very natural parameter, namely the parameter ``size of the instance minus standard parameter''. We showed that such parameterization is able to produce non-trivial parameterized approximation results that, in many cases, can also fit another polynomial-time approximation paradigm: the differential approximation. In any case, studying parameterized approximability of problems with respect to any parameterization under which they are proved to be hard, is very important and adds more and more information about the parameterized intractability of the world of combinatorial problems.

Let us conclude the paper, by saying some words for another equally interesting parameter, the \emph{differential parameter}. It can be defined by $\omega - k$, where~$\omega$ is the worst-case solution value~\cite{approx}. Formally, given an instance~$I$ of a combinatorial problem~$\Pi$,~$\omega(I)$ is the optimal value of a problem~$Pi'$ defined on the same set of instances and having the same feasibility constraints as~$\Pi$, but~$\Pi'$ has the opposite goal. Although for some minimization subset problems, differential and dual parameters coincide (\VC, or \SC, or \DS are such problems), this is not always the case. For \IDS, for example, the worst-solution value on an instance~$I$ is the size of a maximum independent set (that is the largest of the independent dominating sets in~$I$). So, \IDS, while in~\textbf{FPT} when parameterized by the dual parameter (as shown in the beginning of the section), it becomes \textbf{W[1]}-hard when parameterized by the differential parameter~\cite{cross-parameterization}. On the other hand, for many maximization subset problems as, for example, \is, \textsc{max clique}, \textsc{knapsack}, etc., the wost solution (of value~0) is the empty set. There, the differential parameter coincides with the standard one. For \textsc{max minimal vertex cover}, for example, the worst-solution value is the size~$\tau$ of a minimum vertex cover. In any case, a more systematic study of the complexity of exactly or approximately solving problems parameterized by the differential parameter seems to us an interesting direction of future research.

\end{document}